\documentclass[runningheads]{llncs}

\usepackage[T1]{fontenc}
\usepackage[utf8]{inputenc}
\usepackage{amsmath,amssymb,amsfonts,mathtools}
\usepackage{hyperref}
\usepackage{microtype}
\usepackage{enumitem}
\usepackage{booktabs}
\usepackage{array}
\usepackage{tikz}
\usetikzlibrary{arrows.meta,positioning}

\hypersetup{
  colorlinks=true,
  linkcolor=blue,
  citecolor=blue,
  urlcolor=blue
}

\newcommand{\Com}{\mathsf{Com}}
\newcommand{\Vf}{\mathsf{Vf}}
\newcommand{\Adv}{\mathbf{Adv}}
\newcommand{\coeffid}{\text{coeffid}}
\newcommand{\opbind}[2]{\mu_{#1,#2}}

\title{Threshold Authorization Without Threshold Signatures:\texorpdfstring{\\}{ }
Signature-Agnostic MPC Custody}
\titlerunning{Threshold Authorization Without Threshold Signatures}

\author{Dariia Porechna\orcidID{0009-0009-8834-4652}}
\authorrunning{D. Porechna}
\institute{EternaX Labs\\
\email{dariia.p@eternax.ai}}

\begin{document}
\maketitle

\begin{abstract}
Digital-asset custody has been built on threshold multi-party approval: no operation proceeds unless $t$ of $n$ parties approve, and fewer than $t$ compromised parties can neither authorize nor learn the authorization secret.
Threshold signature schemes (TSS) have been the standard mechanism, but the post-quantum transition disrupts this model: standardized hash-based signatures resist efficient threshold signing, and lattice-based threshold protocols remain an emerging research track.

We present a \emph{dual-gate} architecture that separates member authentication from threshold authorization.
Each member signs its approval with an ordinary signature under any EUF-CMA scheme; the quorum jointly produces a \emph{threshold seal} from Shamir-shared secrets bound to the operation.
The seal is the base instance of a programmable authorization computation: simple quorum is the minimal policy, while richer policies can evaluate secret-shared state without making the member-signature scheme part of that computation.
The signature scheme is a deployment parameter: migrating from ECDSA to SLH-DSA or ML-DSA is a key rotation, not a protocol redesign, and members holding keys in commodity HSMs participate through the standard sign API.

The architecture can be deployed wherever the asset-control path supports programmable verification, such as smart contracts, vault modules, or HSMs guarding a master key, and produces an enforcement-layer authorization rather than a native chain signature.
Below-threshold secrecy is information-theoretic; an adversary holding $\geq t$ signing keys but no coefficient shares still cannot produce the seal.

\end{abstract}

\section{Introduction}
\label{sec:intro}

\subsection{The post-quantum custody problem}
\label{subsec:pq-custody}

Digital-asset custody faces the post-quantum transition as a control-plane problem as much as a transaction-signature problem.
A custodian cannot make Bitcoin, Ethereum, or any other chain accept a new signature scheme before the protocol does, but it can make the layers it controls post-quantum ready and cryptographically agile: member authentication, policy enforcement, vendor integrations, backups, and the signing infrastructure that will extend to future on-chain schemes.
That is highlighted in the Taurus June 2026 report~\cite{taurus2026quantum}: readiness means avoiding architecture choices that lock the authorization layer to one classical signature protocol while standards, vendors, and chains are in active research and transition.

The pressure is concrete on both the regulatory and the protocol side.
Executive Order 14412 (June 2026) mandates the transition of federal high-value and high-impact systems to post-quantum digital signatures by the end of 2031 and extends post-quantum FIPS compliance to federal contractors through procurement rules~\cite{eo14412}.
The FIPS-approved hash-based signatures appear on major blockchain migration paths and make the problem more pressing for MPC custody stacks~\cite{ethereum2026pq,taurus2026quantum}.

The cryptographic promise of threshold-signature MPC custody concerns key material: no compromise of fewer than $t$ of the $n$ key shares yields the signing key or a signature.
Threshold ECDSA (GG18/GG20, CMP, FROST~\cite{gennaro2018gg18,canetti2020cmp,komlo2020frost}) delivers this for classical signatures by bundling member authentication and threshold authorization into one cryptographic object.

The post-quantum migration breaks that bundle.
Lattice-based threshold signatures (ML-DSA, Falcon) are an active research and standardization track~\cite{delpino2024thresholdraccoon,delpino2025finally,celi2026efficientmldsa} but not the production baseline occupied by threshold ECDSA.
For hash-based signatures (SLH-DSA~\cite{nist2024fips205}, LMS/HSS), Kondi, Kumar, and Vanegas~\cite{kondi2026blackbox} rule out black-box threshold protocols, so any threshold signing must evaluate the hash function inside MPC at prohibitive cost.
The remaining options are non-black-box MPC for hash evaluation and a trusted setup with large preprocessing state~\cite{kelsey2024distributed}, neither of which are practical product paths for institutional custody.
HSM-based custody can adopt standardized post-quantum signatures as vendor firmware and protocol support arrive; custody built around threshold ECDSA does not have a comparable production path to signature agility.
The question is therefore whether the operational properties of distributed authorization with below-threshold secrecy and share refresh can be preserved for MPC custody without thresholding the signature itself.

\subsection{Separating authentication from authorization}
\label{subsec:separation}

A threshold signature bundles two distinct functions: (a)~member authentication (``who approved'') and (b)~threshold authorization (``did enough parties approve'').
Only~(b) requires the threshold property; only~(a) requires a signature.

Separating the two is an existing custody practice. Production stacks already run an off-chain authorization step ahead of the blockchain signature: in HSM-based dual control, a rule engine verifies each approver's ordinary signature over the transaction intent and, once the policy is satisfied, creates and signs the transaction under a master key~\cite{taurus2026quantum}.
That separation is procedural.
The authorization gate relies on ordinary signature unforgeability plus trusted policy code that counts approvals; it introduces no independent threshold secret hidden from below-threshold coalitions.
Consequently, compromising $t$ approver signing keys is sufficient to pass the gate, and partial key compromise is handled by key rotation rather than share refresh.

The dual gate makes the authorization gate itself cryptographic, backing each of the two functions with the primitive where the required property is native.
Members authenticate with ordinary signatures; threshold authorization comes from a \emph{seal} reconstructed from Shamir-shared secrets, information-theoretically hidden below threshold and proactively refreshable.
This second gate is not merely another signature count: it is a joint evaluation over setup-bound shares, with simple quorum as the degree-1 base case and confidential policy evaluation as the natural extension.
An adversary holding $t$ approver signing keys but not the corresponding coefficient shares passes any signature-counting gate, procedural or on-chain, but cannot produce the seal.

In the classical setting there was no pressure to unbundle the signature itself: threshold ECDSA served custody well, and the bundled signature was the deliverable.
The post-quantum transition changes this.
For hash-based schemes the bundling is precisely what obstructs migration, since the threshold property can no longer be obtained over the signature at practical cost~\cite{kondi2026blackbox}.
Subsequent work has accordingly focused on thresholdizing the signature schemes themselves~\cite{kondi2026blackbox,kelsey2024distributed}.
The dual gate takes the complementary route: the signature is never thresholded, and the threshold property is provided by the gate where it is native and information-theoretic.

What MPC custody provides operationally is an authorization model, distributed authorization with below-threshold secrecy and share refresh, rather than the specific technique of jointly evaluating a signing algorithm.
The dual gate recovers these functional properties without MPC over the signing function, and provides the cryptographic agility for the post-quantum transition: the member-signature scheme is a deployment parameter, so the authorization layer is post-quantum ready ahead of, and independently of, each chain's own signature migration.

\subsection{Contributions}
\label{subsec:contribution}

\begin{enumerate}[leftmargin=*]
  \item \textbf{A decoupled cryptographic authorization gate.}
    A \emph{dual-gate} protocol: an operation passes only through both gates, the \emph{signature check} (each member authenticates with an ordinary signature) and the \emph{seal check} (at least $t$ Shamir-shared evaluations reconstruct a threshold seal bound to the operation).
    Neither gate alone authorizes: stolen member keys satisfy any signature-counting policy but cannot pass the seal check.
  \item \textbf{Ordinary member signatures plus threshold seal.}
    Per-operation approval is non-interactive: each member computes one affine evaluation and one ordinary signature.
    The seal is a Shamir-shared Wegman--Carter one-time authenticator~\cite{wegman1981,krawczyk1994lfsr}; below-threshold secrecy is information-theoretic, and the nontrivial engineering is confined to setup (a JRSS/VSS run).
    This realizes MPC-like authorization semantics without MPC over the signing algorithm: the members jointly evaluate the authorization layer, while signatures only authenticate who contributed.
  \item \textbf{PQ/HSM-friendly migration.}
    The member-signature scheme appears only in the signature check, so it is a deployment parameter.
    Migrating from ECDSA to SLH-DSA or ML-DSA is a key rotation; members in commodity HSMs participate through the standard sign API; mixed-scheme quorums are possible during migration.
  \item \textbf{Programmable enforcement-layer deployment.}
    The construction deploys wherever the asset-control path can run programmable verification logic, such as chain consensus, smart contracts, vaults, account-abstraction layers, or an HSM guarding a master key. The output is an enforcement-layer authorization rather than a native signature.
    The seal generalizes from quorum counting to joint computation of authorization policies over secret-shared state, including confidential policy evaluation (Section~\ref{subsec:programmable}).
\end{enumerate}

\subsection{Context: threshold PQ signatures and multisig}
\label{subsec:related}

The dual gate addresses a gap between two available approaches.

\emph{Threshold PQ signatures.}
Lattice-based threshold signing is an active research track~\cite{delpino2024thresholdraccoon,delpino2025finally,celi2026efficientmldsa} but not yet the production baseline.
For hash-based signatures, black-box threshold protocols are ruled out~\cite{kondi2026blackbox}: threshold signing must evaluate the hash inside MPC.
Kelsey, Lang, and Lucks~\cite{kelsey2024distributed} construct distributed hash-based signatures with a trusted setup or large common reference state, but these still produce threshold signatures; our construction avoids thresholding the member-authentication signature entirely.

\emph{Multisig.}
A $t$-of-$n$ multisig under a post-quantum scheme is deployable wherever the control path can run custom verification logic, whether as smart-account verifier code, future signature precompiles, or the off-chain dual-control pattern above~\cite{ethereum2026pq}.
It provides threshold authorization and attribution but no below-threshold secrecy and no share refresh: a leaked signing key stays leaked until a public key rotation, and compromising $t$ signing keys is total compromise.
Sections~\ref{subsec:multisig} and~\ref{subsec:dualcontrol} detail what the dual gate adds.
The next section presents the construction itself.

\section{The dual-gate architecture}
\label{sec:construction}

The construction instantiates the separation above by requiring an operation $M$ to pass two independent gates over the same canonical message.
The \emph{signature check} verifies that each contribution carries a valid ordinary signature from a registered member.
The \emph{seal check} verifies that at least $t$ setup-bound Shamir shares reconstruct the correct \emph{threshold seal}---the operation-bound evaluation of a jointly held affine map.
Neither gate alone authorizes: signing keys without shares cannot produce the seal; shares without signatures fail the signature check.

The seal is a Shamir-shared instance of the Wegman--Carter one-time authenticator over $\mathbb{F}_p$~\cite{wegman1981,krawczyk1994lfsr}: $\sigma_\nu = k_1 x + k_2$ is a pairwise-independent MAC evaluated at the operation-bound point $x$, which is why coefficient slots are single-use by construction rather than by policy choice (Section~\ref{subsec:slots}).
Figure~\ref{fig:flow} traces one authorization end to end.

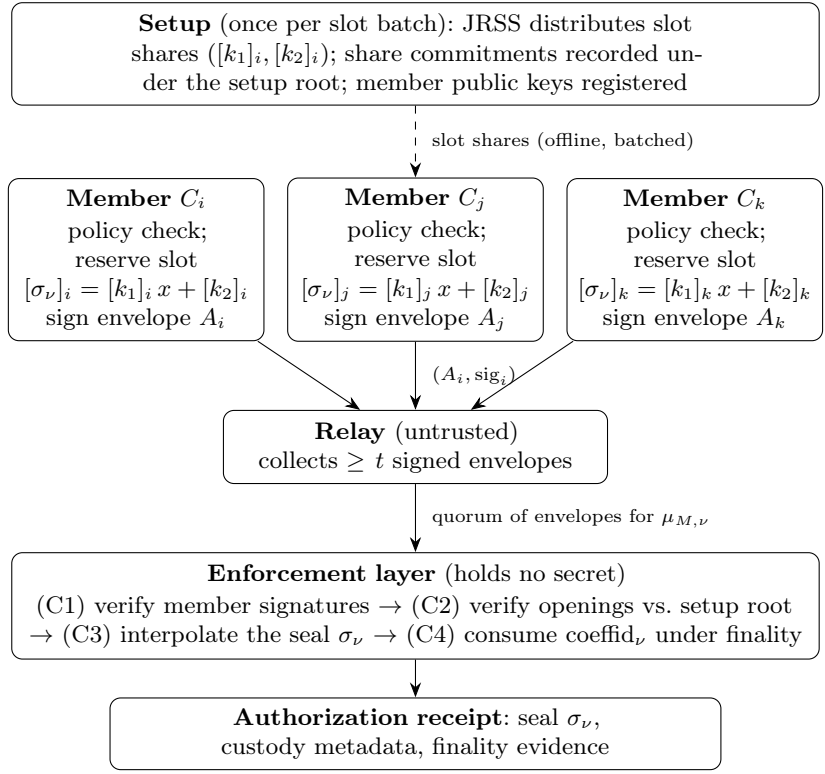
\begin{figure}[t]
\centering
\begin{tikzpicture}[
  font=\footnotesize,
  box/.style={draw, rounded corners, align=center, inner sep=4pt},
  lbl/.style={font=\scriptsize, midway, right=3pt},
  flow/.style={-{Stealth[length=2mm]}}
]
\node[box, text width=10.4cm] (setup) {\textbf{Setup} (once per slot batch):
  JRSS distributes slot shares $([k_1]_i,[k_2]_i)$; share commitments recorded under the setup root; member public keys registered};

\node[box, below=0.95cm of setup, text width=3.1cm] (m2) {\textbf{Member $C_j$}\\[1pt]
  policy check; reserve slot\\
  $[\sigma_\nu]_j = [k_1]_j\,x + [k_2]_j$\\
  sign envelope $A_j$};
\node[box, left=0.3cm of m2, text width=3.1cm] (m1) {\textbf{Member $C_i$}\\[1pt]
  policy check; reserve slot\\
  $[\sigma_\nu]_i = [k_1]_i\,x + [k_2]_i$\\
  sign envelope $A_i$};
\node[box, right=0.3cm of m2, text width=3.1cm] (m3) {\textbf{Member $C_k$}\\[1pt]
  policy check; reserve slot\\
  $[\sigma_\nu]_k = [k_1]_k\,x + [k_2]_k$\\
  sign envelope $A_k$};

\node[box, below=0.9cm of m2, text width=4.8cm] (relay) {\textbf{Relay} (untrusted)\\ collects $\geq t$ signed envelopes};

\node[box, below=0.9cm of relay, text width=10.4cm] (enf) {\textbf{Enforcement layer} (holds no secret)\\[1pt]
  (C1) verify member signatures $\to$ (C2) verify openings vs.\ setup root\\
  $\to$ (C3) interpolate the seal $\sigma_\nu$ $\to$ (C4) consume $\coeffid_\nu$ under finality};

\node[box, below=0.5cm of enf, text width=7.8cm] (rcpt) {\textbf{Authorization receipt}: seal $\sigma_\nu$, custody metadata, finality evidence};

\draw[flow, dashed] (setup) -- node[lbl] {slot shares (offline, batched)} (m2);
\draw[flow] (m1) -- (relay);
\draw[flow] (m2) -- node[lbl] {$(A_i, \text{sig}_i)$} (relay);
\draw[flow] (m3) -- (relay);
\draw[flow] (relay) -- node[lbl] {quorum of envelopes for $\opbind{M}{\nu}$} (enf);
\draw[flow] (enf) -- (rcpt);
\end{tikzpicture}
\caption{One dual-gate authorization end to end.
Members act in parallel with no member-to-member interaction; the relay is untrusted; the enforcement layer verifies both gates and holds no secret.}
\label{fig:flow}
\end{figure}

\subsection{Primitives}
\label{subsec:primitives}

\begin{itemize}[leftmargin=*]
  \item $\mathbb{F}_p$, $p \approx 2^{256}$.
  \item $H \colon \{0,1\}^* \to \mathbb{F}_p$, domain-separated.
  \item $\text{SSS}_{t,n}$: Shamir secret sharing, degree-$(t-1)$ polynomials.
  \item $\text{Sig} = (\text{KeyGen}, \text{Sign}, \text{Verify})$: any EUF-CMA-secure scheme with deterministic verification.
  \item $\Pi = (\Com, \Vf)$: a non-interactive commitment scheme, the \emph{share-correctness profile} (Definition~\ref{def:profile}); the recommended instantiation is the salted-hash commitment $\Com(m;\rho) = H(\text{``custody-commit''} \,\|\, \rho \,\|\, m)$.
\end{itemize}

For the post-quantum profile, $\text{Sig}$ must satisfy EUF-CMA against quantum polynomial-time adversaries with classical signing-oracle access, the standard notion claimed for FIPS~205 (SLH-DSA) and FIPS~204 (ML-DSA).
Where the analysis refers to EUF-CMA without qualification, the applicable notion is that of the deployed scheme.

\subsection{Setup (dealer-free JRSS/VSS)}
\label{subsec:setup}

Setup provisions a batch of $B$ single-use coefficient slots, indexed by $\nu \in [B]$; the steps below run once per slot, with fresh randomness for every slot.

\begin{itemize}[leftmargin=*]
  \item For each slot $\nu$, each $C_j$ samples random $a_j^{(\nu)}, b_j^{(\nu)} \in \mathbb{F}_p$, Shamir-shares each at threshold $t$.
  \item Each $C_i$ sums received shares: $[k_1^{(\nu)}]_i = \sum_j [a_j^{(\nu)}]_i$, $[k_2^{(\nu)}]_i = \sum_j [b_j^{(\nu)}]_i$; the summed sharings define degree-$(t-1)$ polynomials $f_1^{(\nu)}, f_2^{(\nu)}$ with $[k_1^{(\nu)}]_i = f_1^{(\nu)}(i)$, $[k_2^{(\nu)}]_i = f_2^{(\nu)}(i)$.
  \item $k_1^{(\nu)} = \sum_j a_j^{(\nu)}$, $k_2^{(\nu)} = \sum_j b_j^{(\nu)}$ are unknown to any individual party.
  \item For every slot $\nu$ and member $C_i$, the setup records the share commitment $\text{com}_i^{(\nu)} = \Com(([k_1^{(\nu)}]_i, [k_2^{(\nu)}]_i);\, \rho_i^{(\nu)})$ under the setup root; the opening randomness $\rho_i^{(\nu)}$ is held by $C_i$ alongside the shares.
  \item Each member registers a member-authentication public key; the enforcement layer stores the setup transcript root and the selected share-correctness profile.
  \item Production setup must guarantee consistency of the live degree-$(t-1)$ polynomials before coefficients can be activated, and must bind each contributor's randomness before any dealing is revealed (commit-then-reveal); detailed in assumption~S (Section~\ref{subsec:model}).
\end{itemize}

Because every slot uses fresh randomness, the slot polynomials are mutually independent, and exposure of one slot's polynomials reveals nothing about any other.
Batching amortizes setup cost across $B$ operations.
Where the slot is fixed by context we drop the superscript and write $[k_1]_i, [k_2]_i$.

\subsection{Per-operation protocol}
\label{subsec:per-op}

\begin{itemize}[leftmargin=*]
  \item Message binding:
    \[
      \opbind{M}{\nu} = H(\text{``custody-op''} \,\|\, \text{address} \,\|\, \text{policyid} \,\|\, \text{optype} \,\|\, \nu \,\|\, H(M))
    \]
  \item Common scalar:
    \[
      x = H(\text{``custody-affine-x''} \,\|\, \opbind{M}{\nu} \,\|\, \coeffid_\nu)
    \]
  \item Each approving $C_i$:
    \begin{enumerate}
      \item Check local policy and reserve the single-use coefficient slot (signer-side reservation / HSM state), exactly as a presignature is reserved before use.
      \item Evaluate: $[\sigma_\nu]_i = [k_1]_i \cdot x + [k_2]_i$
      \item Form envelope $A_i$ over \text{address}, \text{policyid}, $\coeffid_\nu$, $i$, $\opbind{M}{\nu}$, the evaluation $[\sigma_\nu]_i$, and the share-correctness opening (the revealed slot shares and opening randomness $\rho_i^{(\nu)}$), under a canonical injective encoding.
      \item Sign: $\text{sig}_i = \text{Sig.Sign}(\text{sk}_i, A_i)$
    \end{enumerate}
\end{itemize}

\subsection{Acceptance rule}
\label{subsec:acceptance}

\begin{itemize}[leftmargin=*]
  \item \textbf{(C1) Signature check:} $\exists Q \subseteq \mathcal{C}$, $|Q| \geq t$: all member signatures verify.
  \item \textbf{(C2) Seal check, share correctness:} verify each share via the operation-time opening check: $\Vf(\text{com}_i^{(\nu)}, ([k_1]_i, [k_2]_i), \rho_i^{(\nu)}) = 1$ against the setup root, and check $[\sigma_\nu]_i = [k_1]_i x + [k_2]_i$.
  \item \textbf{(C3) Seal check, reconstruction:} Lagrange-interpolate accepted shares to recover the seal $\sigma_\nu$ (Lemma~\ref{lem:consistency}).
  \item \textbf{(C4) Freshness:} $\coeffid_\nu$ not yet consumed; on success, consume it under finality.
\end{itemize}

\subsection{Programmable authorization}
\label{subsec:programmable}

The affine map of $\sigma_\nu = k_1 x + k_2$ is the degree-1 instance of a more general capability: the seal is the evaluation of a function over secrets fixed at setup and bound to the operation.
Representative policies beyond simple quorum include weighted and hierarchical access structures, risk-dependent authorization (the required threshold changes with transaction value), and confidential running state (cumulative limits enforced without revealing the current balance to any individual party).
The security analysis (Appendix~\ref{app:formal}) treats the affine instance; richer policies inherit the authentication-check and secrecy arguments but require a per-policy correctness analysis.

\section{End-to-end authorization flow}
\label{sec:workflow}

Having specified the cryptographic checks, we next place them in the operational lifecycle of a custody system, demonstrating which parts replace threshold signing and which parts remain ordinary custody maintenance.

\subsection{The custody authorization lifecycle}
\label{subsec:five-ops}

A custody authorization system generally exposes five operations:

\begin{enumerate}[leftmargin=*]
  \item $\text{Setup}(\mathcal{C}, t, n)$: register members, distribute key material, set policy.
    No party learns the full authorization secret.
  \item $\text{Approve}(i, M)$: member $C_i$ produces a contribution for operation $M$.
    Below threshold: zero leakage.
  \item $\text{Authorize}(M, (i, c_i)_{i \in Q}, |Q| \geq t)$: produce a verifiable authorization proof bound to $M$.
  \item $\text{Verify}(M, \pi)$: a programmable enforcement layer checks the proof.
    Verifier holds no secret.
  \item $\text{Maintain}()$: refresh shares, rotate members, update policy without revealing the authorization secret.
\end{enumerate}

These map directly onto Section~\ref{sec:construction}: Setup runs the JRSS sharing and registers member public keys; Approve is the per-operation protocol; Authorize collects $t$ envelopes and runs the acceptance rule; Verify is performed by the enforcement layer; Maintain refreshes coefficient shares or rotates member keys, each itself an authorized custody operation.
The authorization artifact is an enforcement-layer receipt (the reconstructed seal, custody metadata, finality evidence), not a signature under the member scheme.

\subsection{Key material and where it lives}
\label{subsec:key-management}

The construction manages three kinds of material with different lifecycles.

\paragraph{Member signing keys (long-term).}
One keypair per member under any EUF-CMA scheme, used only to sign contribution envelopes.
The key can live in a commodity HSM, an enclave, or a signer service: the protocol needs nothing beyond the standard sign API.
Rotation is a Maintain operation and does not touch the threshold layer.

\paragraph{Coefficient shares (per-slot, single-use).}
Each JRSS batch yields $B$ slots; a member stores its slot shares $([k_1]_i,[k_2]_i)$ and commitment opening randomness alongside a reservation flag.
A slot is consumed by exactly one accepted operation; its shares are dead thereafter.
Proactive refresh re-shares unconsumed slots on a schedule.

\paragraph{Public setup material.}
The setup transcript root, share commitments, member public keys, and policy state are public to the enforcement layer.
The enforcement layer holds no secret and can verify but never produce an authorization.

\subsection{Scope}
\label{subsec:boundary}

The construction replaces the custody authorization workflow wherever the asset-control path can run programmable verification logic and enforce both gates.
Smart-contract deployments implement both gates directly.
It does not produce a native threshold signature under any member scheme, and does not transparently control assets whose only control path is a legacy native-account signature verifier.

\section{Deployment: signature agility and HSMs}
\label{sec:deployment}

The lifecycle view isolates the signature primitive from the threshold layer.
This section records the resulting deployment consequences, especially for post-quantum migration and HSM-based custody.

\subsection{The signature scheme is a deployment parameter}
\label{subsec:agility}

The seal's field arithmetic, Shamir sharing, and reconstruction are defined over $\mathbb{F}_p$ with zero coupling to the member-signature primitive.
This yields the cryptographic agility the post-quantum transition demands.

\paragraph{Scheme choice.}
The construction works with SLH-DSA, ML-DSA, ECDSA, EdDSA, or any future EUF-CMA scheme.

\paragraph{Migration is key rotation.}
Moving a member from ECDSA to SLH-DSA is a Maintain operation (register the new public key, retire the old); the threshold layer, policy engine, and audit pipeline are untouched.

\paragraph{Hybrid quorums.}
The verifier checks each envelope under that member's registered key, so members on different schemes can coexist in one quorum, permitting a staged migration.

\paragraph{Post-quantum posture.}
The seal relies only on hash commitments and information-theoretic Shamir secrecy; once every member key is post-quantum, the whole deployment is, with no discrete-logarithm or pairing assumptions.

The security analysis (Appendix~\ref{app:formal}) consumes only the EUF-CMA advantage of each registered key.
A mixed quorum inherits the weakest deployed scheme: an operation approved by a quorum containing even one classical key carries the classical bound for that operation.
By contrast, changing the signature scheme under a threshold signature protocol means a new DKG, a new presignature lifecycle, a new combination algorithm, and a new security analysis.

\subsection{Working with HSMs}
\label{subsec:hsm}

The construction naturally composes with HSM custody.

\paragraph{Member keys in commodity HSMs.}
Member signing keys can live in commodity HSMs: the protocol requires only an ordinary signature over the contribution envelope, so any device exposing a standard sign API for a supported scheme participates directly, and as vendor firmware ships FIPS 204/205 support, HSM-resident post-quantum member keys work without protocol changes.
This differs from threshold signing, where each member's contribution is a scheme-specific partial-signature computation that commodity HSM firmware does not implement.

\paragraph{Custody state beyond the key.}
Standard HSM firmware protects the \emph{signing key} only: coefficient shares, opening randomness, and reservation state live in the member's signer service or in a device with custom firmware.
This storage separation is a feature of the deployment (Section~\ref{subsec:multisig}), but it must be provisioned deliberately.

\paragraph{The HSM as enforcement layer.}
The enforcement layer itself can be an HSM-hosted policy module: it needs signature verification, field arithmetic, and a hash, but no secret key material.

\paragraph{Master-key custody.}
Where an asset master key is held inside one HSM, unlocked by a quorum of approver signatures~\cite{taurus2026quantum}, the HSM \emph{is} the enforcement layer.
At setup it stores the setup root and registered member keys; at operation time it runs C1--C4 and releases the master-key signature only when the seal verifies.
This upgrades the unlock condition from ``count $t$ approver signatures'' to the two-factor quorum of Section~\ref{subsec:multisig}, and the output is a \emph{native} signature under the master key, recovering drop-in compatibility with legacy chains at the cost of the master key's hardware trust anchor.

The practical consequence for the HSM-vs-MPC framing~\cite{taurus2026quantum}: distributed $t$-of-$n$ authorization with below-threshold secrecy no longer requires giving up HSM-resident, standardized signature keys.

\subsection{Enforcement-layer requirements}
\label{subsec:enforcement}

The construction requires the enforcement layer to provide the following components:
(1)~signature verification for the signature check,
(2)~$\mathbb{F}_p$ arithmetic and a hash for the seal check,
(3)~an operation-time share-correctness opening check (hash-committed openings need only a collision-resistant hash; no pairings, no lattice arithmetic, no discrete-log assumption),
(4)~ordered finality for single-use coefficient-id consumption (Section~\ref{subsec:slots}),
(5)~policy-state binding.

\subsection{Deployment surfaces}
\label{subsec:surfaces}

The target is rare, high-value custody operations: withdrawals, mint/burn, key rotation, delegation.
Both gates can serve as consensus-level validity conditions on a native chain, or as acceptance logic in a smart-contract custody module.
Gas cost is the practical boundary for on-chain deployment.

\section{Security rationale}
\label{sec:security}

The deployment claims above rely on the two gates composing in the intended way: signatures provide attribution, while the seal provides threshold authorization and secrecy.
The formal scheme, share-correctness profile definitions, security experiments, theorem statements, and all proofs are collected in Appendices~\ref{app:formal} and~\ref{app:proofs}.
This section states the system model, the custody guarantees, and the attack matrix that deployments should evaluate.

\subsection{System and adversary model}
\label{subsec:model}

\paragraph{Parties.}
$n$ custody members $C_1,\dots,C_n$, each holding a long-term member signing key and per-slot coefficient shares with their opening randomness; an untrusted relay that proposes operations and transports contribution envelopes; and an enforcement layer (chain consensus, contract, vault, or issuer module) that executes C1--C4 over public state.

\paragraph{Enforcement layer.}
Trusted for \emph{integrity only}: it executes C1--C4 correctly and consumes coefficient identifiers atomically in a total order (ordered finality).
It holds no custody secret.
A compromised verifier learns nothing beyond public data and can verify the seal but never produce one (Remark~\ref{rem:verifier}).

\paragraph{Relay and network.}
The relay and network are adversarial: the adversary observes, delays, drops, reorders, and injects envelopes at will.
Withholding envelopes affects liveness only; the guarantees below establish safety.

\paragraph{Corruption model.}
The adversary corrupts members statically, before setup, at two grades:
\begin{itemize}[leftmargin=*]
  \item \emph{Key-only corruption} ($i \in K$): the adversary learns the signing key $\text{sk}_i$; the member itself continues to guard its slot shares and follow signer-side single-use reservation.
  \item \emph{Full corruption} ($i \in F$): the adversary learns $\text{sk}_i$ and the member's entire slot state and acts arbitrarily on its behalf.
\end{itemize}
$K$ and $F$ are disjoint subsets of $[n]$.
The custody threshold assumption is that fewer than $t$ members are fully compromised; the formal analysis quantifies exactly what mixed corruptions ($K \neq \emptyset$) can and cannot do.

\paragraph{Setup assumption (S).}
Setup is modeled as an ideal sharing functionality: for every slot $\nu$ it fixes degree-$(t-1)$ polynomials whose constant terms are uniform and independent across slots, delivers shares and commitment randomness to each member, and publishes commitments under the setup root.
The production JRSS must satisfy: (a)~at least one contributor is honest; (b)~contributor randomness is bound before reveal (commit-then-reveal); (c)~the VSS sub-protocols enforce agreement on the dealt polynomials; (d)~share commitments are derived deterministically from the transcript, so all parties compute the same root.

\subsection{Custody guarantees}
\label{subsec:invariants}

To provide the expected custody functionality, the construction should provide the following guarantees:

\begin{itemize}[leftmargin=*]
  \item \textbf{(I0) Threshold authorization.}
    No set of fewer than $t$ effectively corrupted custody parties can produce an accepted authorization (Theorem~\ref{thm:uf}).
  \item \textbf{(I1) Threshold secrecy.}
    Below-threshold corruption yields no information about the secrets of any unused coefficient slot (Theorem~\ref{thm:priv}).
  \item \textbf{(I2) Per-operation binding.}
    An accepted authorization is bound to $M$ and to one coefficient slot, and cannot be retargeted or replayed (Proposition~\ref{prop:binding}).
  \item \textbf{(I3) Member attribution.}
    Each accepted contribution is attributable to a registered member, and honest members cannot be framed (Proposition~\ref{prop:attr}).
\end{itemize}

Additionally, the deployment should ensure that the following conditions are met:

\begin{itemize}[leftmargin=*]
  \item \textbf{(D1) Share correctness.}
    The share-correctness profile must be hiding and opening-unforgeable (Definition~\ref{def:profile}); the recommended salted-hash profile satisfies both (Lemma~\ref{lem:profiles}).
  \item \textbf{(D2) Single-use coefficient slots.}
    Signer-side reservation (each honest member evaluates each slot for at most one message) and enforcement-side consumption of $\coeffid_\nu$ under ordered finality.
\end{itemize}

The end-to-end security level is $\min$(hash security, member-signature EUF-CMA).
The formal statements and reductions appear in Appendix~\ref{app:formal}; all proofs in Appendix~\ref{app:proofs}.

\subsection{Attack matrix}
\label{subsec:composition}

Table~\ref{tab:attack-matrix} instantiates the formal results against representative adversaries.
The custody threshold is the security boundary: the signature check prevents attribution and relay-substitution attacks; the seal check prevents a below-threshold set of valid member identities from authorizing.

\begin{table}[t]
\caption{Dual-gate composition: representative adversaries.}
\label{tab:attack-matrix}
\centering
\small
\begin{tabular}{@{}>{\raggedright\arraybackslash}p{3.6cm}c>{\raggedright\arraybackslash}p{5.8cm}@{}}
\toprule
\textbf{Attacker capability} & \textbf{Authorize?} & \textbf{Why} \\
\midrule
Steal one member's $\text{sk}_i$ only &
No &
Still needs that member's slot opening and $t-1$ additional signed contributions (Theorem~\ref{thm:uf}) \\
Compromise $< t$ custody parties &
No &
Below-threshold shares give zero information about coefficients (Theorem~\ref{thm:priv}) and cannot reconstruct the seal \\
Steal $\geq t$ member signing keys but not coefficient shares &
No\textsuperscript{$\dagger$} &
Operation-time openings bind each share to setup; a fabricated opening forges the profile commitment (Theorem~\ref{thm:uf}). \textsuperscript{$\dagger$}Per slot: holds for slots with no prior approval by a stolen-key member (Remark~\ref{rem:stolen-keys}) \\
Compromise $\geq t$ custody parties &
Yes &
They can sign envelopes and evaluate the affine map; this is the intended custody threshold ($|F| \geq t$) \\
Compromise validators only &
No &
Validators verify both gates but hold no custody coefficient shares (Remark~\ref{rem:verifier}) \\
\bottomrule
\end{tabular}
\end{table}

\begin{remark}[The stolen-keys guarantee is per-slot]
\label{rem:stolen-keys}
The third row's ``No'' is scoped: a key-compromised member that has approved \emph{any} message on slot $\nu$ has published its slot-$\nu$ opening, and from that point only its signature---which the adversary holds---protected its slot-$\nu$ contribution.
If some $t$ stolen-key members have all previously approved on one unconsumed slot, the adversary can re-sign the published openings under fresh envelopes for a message of its choice.
The row's guarantee therefore covers slots on which no stolen-key member has approved, which includes every slot of a fresh batch; consumed slots are closed by C4 regardless.
Operationally, discovered key compromise should trigger an immediate freeze of unconsumed slots carrying prior approvals from the compromised members.
\end{remark}

\begin{remark}[Verifier independence]
\label{rem:verifier}
The enforcement layer's state is public and its algorithm is deterministic; a compromised verifier learns nothing beyond public data.
What verifier compromise can do is violate the integrity assumption itself---accepting without running C1--C4---which is the standing assumption on the layer that ultimately moves the asset, identical for any custody design including threshold signatures.
\end{remark}

\section{Comparison with deployment baselines}
\label{sec:comparison}

The comparison below clarifies what the presented construction changes relative to the nearest deployment baselines: threshold signing, on-chain multisig, and off-chain dual control.

\subsection{Workflow mapping}
\label{subsec:workflow-table}

Table~\ref{tab:workflow-map} maps each custody lifecycle operation to threshold ECDSA and the dual gate.
The key structural difference: migrating the signature primitive under threshold ECDSA means rebuilding the entire threshold protocol, DKG profile, presignature lifecycle, and integration surface.
Under the dual gate, migration is a key rotation.

\begin{table}[t]
\caption{Threshold ECDSA vs.\ dual gate across the custody lifecycle.}
\label{tab:workflow-map}
\centering
\small
\begin{tabular}{@{}p{1.35cm}>{\raggedright\arraybackslash}p{4.1cm}>{\raggedright\arraybackslash}p{4.4cm}@{}}
\toprule
\textbf{Op} & \textbf{Threshold ECDSA} & \textbf{Dual gate} \\
\midrule
Setup &
DKG distributes ECDSA key shares &
VSS/AVSS/JRSS distributes affine coefficient shares; each member registers a signature public key \\
Approve &
Produce partial ECDSA signature share &
Produce a signed contribution envelope: member signature over $(\opbind{M}{\nu}, \coeffid_\nu, i, [\sigma_\nu]_i)$ and the share opening \\
Authorize &
Combine partial signatures into one ECDSA signature &
Reconstruct the seal $\sigma_\nu$ from authenticated evaluation shares after opening verification \\
Verify &
Verify the combined ECDSA signature against the public key &
Verify member signatures (C1) and the reconstructed seal (C2--C3) \\
Maintain &
Proactive key-share refresh, re-sharing for member rotation &
Proactive coefficient refresh, member key rotation, policy update (each itself an authorized custody op) \\
\midrule
Output &
Native threshold ECDSA signature verifiable by any ECDSA checker &
Enforcement-layer authorization receipt; not a native signature, unless the enforcement layer is an HSM guarding a master key (Section~\ref{subsec:hsm}) \\
\bottomrule
\end{tabular}
\end{table}

\subsection{On-chain multisig as a baseline}
\label{subsec:multisig}

A $t$-of-$n$ multisig under a post-quantum scheme is the simplest programmable baseline: register $n$ member keys and have the verifier count $t$ valid signatures.
On Ethereum today this means a smart-account or account-abstraction path until future signature precompiles or native signature verification are deployed; deployment is limited by gas cost, signature size, and the keys that control upgrades or account recovery~\cite{ethereum2026pq}.
It provides threshold authorization, per-operation binding, and attribution, and no threshold-signing problem arises because nothing is thresholded.
The dual gate keeps that simplicity and adds what a signature count cannot provide.

\paragraph{A second, independent authorization factor.}
In multisig, compromising $t$ signing keys is total compromise.
In the dual gate, an adversary holding $\geq t$ stolen signing keys still fails the seal check (Table~\ref{tab:attack-matrix}): coefficient shares are separate material, typically in separate storage, with information-theoretic below-threshold secrecy.

\paragraph{Confidential, programmable policy.}
Multisig policy is public verifier code over public identities.
The seal evaluates policy over state that is itself secret-shared, which no single member or verifier can read or unilaterally alter (Section~\ref{subsec:programmable}).

\paragraph{Below-threshold recovery.}
Proactive refresh re-randomizes coefficient shares, so below-threshold share exfiltration is neutralized by the next refresh.
Multisig has no analogous recovery: a leaked signing key stays leaked until a public key rotation.

\subsection{Off-chain dual control as a baseline}
\label{subsec:dualcontrol}

The same quorum logic is implemented today as an off-chain product pattern: a rule engine, typically hosted in an HSM, verifies each approver's ordinary signature over the transaction intent and, once $t$ valid approvals are counted, creates and signs the blockchain transaction under a master key~\cite{taurus2026quantum}.
Cryptographically this is the multisig of Section~\ref{subsec:multisig} relocated behind a trust anchor, and it inherits the same limits: compromising $t$ approver keys is total compromise, the policy is trusted code over public identities, and there is no below-threshold secrecy or share refresh.
It also concentrates trust beyond what on-chain multisig requires: the rule engine and the master key form a single point whose compromise bypasses the quorum entirely.

The dual gate composes with this deployment rather than competing with it.
Where the control path is programmable, the enforcement layer runs both gates and holds no secret at all.
Where a master key deployment is preferable, Section~\ref{subsec:hsm} keeps the HSM as the enforcement layer but upgrades its unlock condition from ``count $t$ approver signatures'' to the two-factor quorum: an adversary holding $t$ approver keys still satisfies the rule engine's signature count but cannot produce the seal.

\subsection{Strengths}
\label{subsec:strengths}

\paragraph{Cryptographically simpler.}
No curve arithmetic, no multi-round signing MPC, no zero-knowledge machinery in the signing path: one field evaluation and one ordinary signature per member, verified by hashing and interpolation.

\paragraph{More efficient per operation.}
Approval is non-interactive given a slot; the online phase is one round of envelope collection.

\paragraph{Signature-scheme agnostic.}
Any EUF-CMA scheme, including hash-based schemes for which threshold signing is prohibitively expensive~\cite{kondi2026blackbox}.

\paragraph{No long-lived joint algebraic key.}
Threshold ECDSA publishes a group public key whose discrete log is the joint signing key, a standing target for a future quantum adversary.
The dual gate publishes no joint key; member keys are individually replaceable and below-threshold secrecy of the seal coefficients is information-theoretic.

\paragraph{HSM-compatible.}
Members and verifiers need only standard primitives (Section~\ref{subsec:hsm}).

\section{Operational requirements and limitations}
\label{sec:operational}

Separating authorization from native threshold signing moves some costs into operational deployment and enforcement instead of absorbing them inside the signature abstraction.

\subsection{Limitations}
\label{subsec:limitations}

\paragraph{Not a signature.}
The output is a seal plus metadata, verifiable by an enforcement layer that implements both gates; TSS outputs a native signature any standard verifier accepts.
Assets controlled solely by a legacy native-account verifier are out of scope unless the enforcement layer is an HSM guarding the account's master key (Section~\ref{subsec:hsm}).

\paragraph{Custom verifier logic.}
The enforcement layer must implement the acceptance rule (Section~\ref{subsec:enforcement}); TSS verification is a stock signature check.

\paragraph{Per-operation state.}
Single-use coefficient slots must be provisioned and their consumption tracked under finality---the same discipline as presignatures, but it is state a stateless-verification design would not carry.

\paragraph{Quorum visibility.}
Accepted envelopes identify the approving members to the enforcement layer.
Attribution is a custody feature, but a combined threshold signature can hide the approving quorum from the verifier.

\paragraph{Setup requirements.}
Production deployments need a VSS/AVSS/JRSS setup that exports the per-share commitment material the opening check consumes.

\paragraph{Economic quorum capture.}
Bribing or coercing $t$ members defeats the construction exactly as it defeats multisig and threshold signing; the custody threshold is an assumption about member independence that no authorization cryptography replaces.

\subsection{Single-use coefficient slots}
\label{subsec:slots}

A coefficient slot is the analog of a presignature in production threshold ECDSA and must be consumed exactly once.
Enforcement is two-layered: enforcement-layer $\coeffid$ consumption under ordered finality prevents a second \emph{accepted} operation, and signer-side reservation prevents a policy-conforming signer from emitting two evaluations of one slot.
Reuse burns the affected slot's ephemeral authorizer only; unlike ECDSA nonce reuse it does not leak any long-term key.
Batch pre-sharing amortizes setup cost for frequent operations; just-in-time sharing suits rare, high-value custody ops.
The batch cost is modest: the JRSS round count is independent of $B$ (one commit round, one dealing round, and the consistency rounds of the chosen VSS), and bandwidth is linear in $B$; at $n = 20$, $B = 10^4$, and 32-byte field elements a member transmits ${\sim}12$\,MB per batch, an unremarkable offline run.

Ordered finality requires that for each slot $\nu$, the enforcement layer produces exactly one accept decision, atomic with the consumption of $\coeffid_\nu$, in one total order visible to every system that acts on authorizations.
On multi-shard or rollup architectures, slot consumption must live at one serialization point common to all surfaces that can accept operations.
Slot exhaustion is a liveness consideration, never a safety one: an empty batch stalls new authorizations until the next JRSS run and forfeits nothing.

\subsection{Setup root management}
\label{subsec:root}

The setup transcript root anchors every share-correctness opening.
It is authenticated to the enforcement layer at deployment: a chain state variable, a contract storage variable, or HSM state provisioned at commissioning.
Exactly one root is live per custody wallet.
Proactive refresh re-shares unconsumed slots and re-commits them under a new root; the rotation is itself an authorized custody operation, consuming a slot under the old root.

\subsection{Open questions}
\label{subsec:open}

\begin{itemize}[leftmargin=*]
  \item \textbf{Refresh and root rotation:} how refreshed shares re-commit under a new root without an authorization gap, including dual-root acceptance windows and cross-root replay analysis.
  \item \textbf{Slot management under asynchrony:} availability tradeoffs of signer-side reservation versus just-in-time setup, particularly for small thresholds.
  \item \textbf{Setup instantiation:} which VSS/AVSS/JRSS protocols export the per-share verification material the hash-commitment profile consumes. Natural candidates are the Cachin--Kursawe--Lysyanskaya--Shoup AVSS framework~\cite{cachin2002avss} for post-quantum deployments and KZG polynomial-commitment AVSS~\cite{kate2010kzg} for classical; verifying a specific production protocol against assumption~S conditions (a)--(d) remains open.
  \item \textbf{Adaptive corruption:} the model is static; lifting to adaptive corruption faces the usual commit-and-guess obstacles.
\end{itemize}

\section{Reference implementation}
\label{sec:implementation}

To validate that these requirements are implementable with standard components, we include a compact Rust demonstrator of the full dual-gate protocol with a pluggable share-correctness profile.
The implementation uses the secp256k1 scalar field for seal-check arithmetic and Ed25519 for the signature check; both are classical stand-ins chosen for demonstrator convenience.
A production post-quantum profile pairs a post-quantum member scheme (SLH-DSA or ML-DSA) with the hash-commitment profile.

Three share-correctness profiles are implemented:
the \textit{hash-committed opening profile} (recommended, post-quantum), the \textit{secp256k1 Pedersen commitment profile} (classical discrete-log demonstrator), and an \textit{external-AVSS-transcript adapter} for existing setups that do not export public operation-time openings.
The test suite covers Shamir roundtrips, happy-path custody (2-of-3, 3-of-5), quorum consistency, message binding, rejection of the keys-without-shares adversary under the hash-commitment profile, rejection of tampered evaluations, and standard rejection cases.

The reference implementation is publicly available~\cite{dualgateimpl}.

\section{Conclusion}
\label{sec:conclusion}

Digital asset custody needs threshold multi-party approval, and threshold signatures have been the default way to get it. 
The post-quantum transition changes the price on that default: for standardized hash-based schemes, any threshold signing protocol must evaluate the hash function inside MPC at prohibitive cost~\cite{kondi2026blackbox}.
The dual gate does not pay that price, because it never thresholds the signature: members sign ordinarily, and the threshold property is provided by a seal jointly computed from Shamir-shared secrets, where it is native and information-theoretic.
The result is custody authorization from standard primitives in which the signature scheme is a deployment parameter: post-quantum migration is a key rotation, hybrid quorums need no flag day, and the same protocol runs at consensus level, in a smart contract, or inside an HSM guarding a master key.
What a threshold signature provides and the dual gate does not is a native signature any stock verifier accepts; where an enforcement layer can check two gates instead of one signature, the migration becomes routine.

\paragraph{Acknowledgments.}
The author thanks Paarrthhh Sharad Birla (EternaX Labs), whose early identification of the Taurus report~\cite{taurus2026quantum} and the post-quantum gap in MPC custody prompted this line of work; Jean-Philippe Aumasson (Taurus) for feedback on the post-quantum custody landscape and the practical deployment considerations; and Chen Feng (University of British Columbia) for detailed discussions on the security model and the relationship to threshold and MPC primitives.

\bibliographystyle{splncs04}
\bibliography{references}

\appendix
\renewcommand{\theHsection}{appendix.\arabic{section}}

\section{Formal security model and results}
\label{app:formal}

For completeness, this appendix records the formal scheme definition, share-correctness profile, security experiments, and all theorem statements supporting Section~\ref{sec:security}.
All proofs are collected in Appendix~\ref{app:proofs}.
Throughout, $\lambda$ is the security parameter, adversaries are probabilistic polynomial time (PPT), members are indexed by $i \in [n]$, and coefficient slots by $\nu \in [B]$.

\subsection{The share-correctness profile}

\begin{definition}[Share-correctness profile]
\label{def:profile}
A share-correctness profile is a non-interactive commitment scheme $\Pi = (\Com, \Vf)$ over message space $\mathbb{F}_p^2$ and randomness space $R_\lambda$: $\Com(m; \rho)$ outputs a commitment, $\Vf(\text{com}, m, \rho) \in \{0,1\}$ is deterministic, and $\Vf(\Com(m;\rho), m, \rho) = 1$ for all $m, \rho$.
The profile must satisfy:
\begin{itemize}[leftmargin=*]
  \item \textbf{Hiding.} For every PPT $\mathcal{D}$ choosing $m_0, m_1$,
    \[
      \Adv^{\mathrm{hide}}_\Pi(\mathcal{D}) = \bigl| \Pr[\mathcal{D}(\Com(m_0;\rho)) = 1] - \Pr[\mathcal{D}(\Com(m_1;\rho)) = 1] \bigr|,
      \quad \rho \leftarrow R_\lambda,
    \]
    is negligible.
  \item \textbf{Opening unforgeability (OU).} For every PPT $\mathcal{A}$,
    \[
      \begin{aligned}
        \Adv^{\mathrm{ou}}_\Pi(\mathcal{A}) ={} &
        \Pr\bigl[ m \leftarrow \mathcal{A};\ \rho \leftarrow R_\lambda; \\
        & (m', \rho') \leftarrow \mathcal{A}(\Com(m;\rho)) \,:\,
          \Vf(\Com(m;\rho), m', \rho') = 1 \bigr]
      \end{aligned}
    \]
    is negligible.
\end{itemize}
\end{definition}

OU is the binding-style property the acceptance rule actually consumes: a commitment cannot be opened \emph{at all} without its opening information, even by an adversary that already knows the committed message.
This is what defeats the keys-without-shares adversary.

\subsection{Deployed profiles}

\begin{lemma}[Deployed profiles]
\label{lem:profiles}
(i)~The salted-hash profile, with $\rho \leftarrow \{0,1\}^\lambda$ and $\ell$-bit output,
\[
  \Com(m;\rho) = H(\text{``custody-commit''} \,\|\, \rho \,\|\, m),
\]
satisfies, with $H$ modeled as a random oracle that is only observed and never programmed,
\[
  \Adv^{\mathrm{hide}} \leq q_H\,2^{-\lambda}
  \;\text{ and }\;
  \Adv^{\mathrm{ou}} \leq q_H\,2^{-\lambda} + (q_H + 1)\,2^{-\ell}
\]
for adversaries making $q_H$ oracle queries.
(ii)~The Pedersen profile in a prime-order group with independent generators,
\[
  \Com((m_1, m_2); \rho) = g_1^{m_1} g_2^{m_2} h^{\rho},
\]
is perfectly hiding, and $\Adv^{\mathrm{ou}}$ is bounded by the discrete-log advantage in that group.
\end{lemma}

\begin{proof}
(i)~Hiding: the commitment is the oracle's output on an input containing the uniform salt $\rho$; the distributions for $m_0$ and $m_1$ differ only if the distinguisher queries an input with prefix $\text{``custody-commit''} \| \rho$, an event of probability at most $q_H 2^{-\lambda}$.
OU: a valid opening satisfies $H(\text{``custody-commit''} \| \rho' \| m') = \text{com}$; the adversary either queries the challenge preimage itself (probability at most $q_H 2^{-\lambda}$ over $\rho$), or one of its other queries or its final unqueried output collides with the $\ell$-bit value $\text{com}$ (probability at most $(q_H+1) 2^{-\ell}$).
(ii)~Perfect hiding is standard ($h^\rho$ is uniform).
For OU, a valid opening $(m', \rho')$ satisfies $g_1^{m'_1} g_2^{m'_2} h^{\rho'} = g_1^{m_1} g_2^{m_2} h^{\rho}$: if $(m', \rho') = (m, \rho)$ the adversary has computed the discrete logarithm $\rho$ of $\text{com} \cdot g_1^{-m_1} g_2^{-m_2}$ to base $h$; otherwise it has produced a nontrivial discrete-logarithm relation among $g_1, g_2, h$.
\end{proof}

\begin{remark}[Quantum adversaries and the QROM]
\label{rem:qrom}
Lemma~\ref{lem:profiles}(i) is stated in the classical random-oracle model.
A quantum adversary evaluates $H$ in superposition, so the post-quantum profile must be assessed in the quantum random-oracle model (QROM)~\cite{boneh2011qrom}.
Both properties lift with the generic square-root degradation.
Hiding: by the one-way-to-hiding lemma~\cite{ambainis2019o2h}, $\Adv^{\mathrm{hide}} = O(q_H\,2^{-\lambda/2})$.
OU: by the average-case search bounds of~\cite{hulsing2016multitarget}, $\Adv^{\mathrm{ou}} = O(q_H^2\,2^{-\ell} + q_H^2\,2^{-\lambda})$, and the multi-target variant across all $nB$ commitments is $O(q_H^2\,nB\,2^{-\ell})$.
At $\lambda = \ell = 256$ the single-target bounds retain ${\approx}128$ bits of quantum security, and the multi-target bound stays above $100$ bits even for $nB = 2^{40}$ slots.
The collision advantage $\Adv^{\mathrm{cr}}_H$ of Theorem~\ref{thm:uf} is $O(q_H^3\,2^{-\ell})$ in the QROM~\cite{zhandry2015collision}; a deployment wanting this term at $128$ bits instantiates the operation-binding hash at $\ell = 384$.
Because every reduction in this paper only observes the oracle and never programs or rewinds it, the lifting is a term-by-term replacement of the classical lemmas by these QROM bounds.
\end{remark}

\begin{remark}[The salt is load-bearing]
\label{rem:salt}
A deterministic hash commitment over uniform shares is hiding and one-way in the random-oracle model,
but it is \emph{not} opening-unforgeable against an adversary that knows the shares:
whoever reconstructs a member's shares can recompute the commitment and open it.
That adversary is reachable: $t-1$ fully corrupted shares plus one honest approval of a \emph{different} message on the same slot determine the slot polynomials.
The per-share salt $\rho$ closes this: knowing the shares does not yield an opening,
so cross-message exposure of a slot never escalates beyond the members whose own keys are stolen.
\end{remark}

\subsection{Scheme and acceptance}

Write $\mathsf{DG} = (\mathsf{Setup}, \mathsf{Approve}, \mathsf{Accept})$ for the construction of Section~\ref{sec:construction}: $\mathsf{Setup}$ is the ideal functionality of assumption S together with member key registration; $\mathsf{Approve}(i, M, \nu)$ produces member $i$'s signed envelope per Section~\ref{subsec:per-op}; $\mathsf{Accept}$ is the stateful acceptance rule of Section~\ref{subsec:acceptance}.
On submission of an operation $(M, \nu)$ with envelopes $\{(A_i, \text{sig}_i)\}_{i \in Q}$ for a set $Q$ of distinct member indices, $|Q| \geq t$, $\mathsf{Accept}$ recomputes $\opbind{M}{\nu}$ and $x$ and checks
C1: each $\text{sig}_i$ verifies under the registered $\text{pk}_i$ and the fields of $A_i$ match $(\opbind{M}{\nu}, \coeffid_\nu, i)$;
C2: $\Vf(\text{com}_i^{(\nu)}, s'_i, \rho'_i) = 1$ for the opened shares $s'_i = (k'_{1,i}, k'_{2,i})$, and $[\sigma_\nu]_i = k'_{1,i}\, x + k'_{2,i}$;
C3: interpolate $\sigma_\nu$;
C4: $\coeffid_\nu \notin \mathsf{Consumed}$, and on acceptance add it.

\begin{lemma}[Quorum consistency]
\label{lem:consistency}
Under assumption S, honest envelopes for the same $(M, \nu)$ from any two quorums $Q, Q' \subseteq [n]$ of size $t$ are accepted and reconstruct the same seal $\sigma_\nu = k_1^{(\nu)} x + k_2^{(\nu)}$.
\end{lemma}

\begin{proof}
Honest evaluations are values $g(i) = x f_1^{(\nu)}(i) + f_2^{(\nu)}(i)$ of the polynomial $g = x f_1^{(\nu)} + f_2^{(\nu)}$ of degree at most $t-1$; any $t$ distinct points determine $g$, and C3 returns $g(0) = x k_1^{(\nu)} + k_2^{(\nu)}$ for either quorum.
\end{proof}

\subsection{Security experiments}

All experiments share the following environment.
The adversary $\mathcal{A}$ first outputs parameters $(t, n, B)$ and disjoint corruption sets $K, F \subseteq [n]$.
The challenger runs $\mathsf{Setup}$, hands $\mathcal{A}$ the public state, the signing keys $\{\text{sk}_i\}_{i \in K \cup F}$, and the full slot state $\{(s_i^{(\nu)}, \rho_i^{(\nu)})\}_{i \in F,\, \nu \in [B]}$, and then answers two oracles:
\begin{itemize}[leftmargin=*]
  \item $\mathcal{O}_{\mathsf{Approve}}(i, M, \nu)$ for $i \notin F$: returns $\mathsf{Approve}(i, M, \nu)$, which reveals $(s_i^{(\nu)}, \rho_i^{(\nu)})$ by design.
    At most one query per pair $(i, \nu)$ is answered (D2).
  \item $\mathcal{O}_{\mathsf{Accept}}(\cdot)$: runs the stateful $\mathsf{Accept}$ on an adversarial submission, returns the result, and updates $\mathsf{Consumed}$.
\end{itemize}

\begin{definition}[Effective approver set]
\label{def:effective}
For a message $M$ and slot $\nu$, the \emph{effective approver set} at any point of an experiment is
\[
  \begin{aligned}
    G(M, \nu) ={} & F \\
    & \cup \{\, i \in K : \mathcal{O}_{\mathsf{Approve}}(i, \cdot, \nu) \text{ was queried} \,\} \\
    & \cup \{\, i \notin K \cup F : \mathcal{O}_{\mathsf{Approve}}(i, M, \nu) \text{ was queried} \,\}.
  \end{aligned}
\]
\end{definition}

A fully corrupted member effectively approves everything; an honest member effectively approves exactly what it signed; and a key-compromised member effectively approves anything on slot $\nu$ once it has exposed its slot-$\nu$ opening by approving some message on that slot.

\begin{definition}[Security experiments]
\label{def:experiments}
\begin{itemize}[leftmargin=*]
  \item \textbf{Threshold unforgeability ($\mathrm{TA}$-$\mathrm{UF}$).}
    $\mathcal{A}$ wins if some $\mathcal{O}_{\mathsf{Accept}}$ query accepts an operation $(M^*, \nu^*)$ while $|G(M^*, \nu^*)| < t$.
    $\Adv^{\mathrm{uf}}(\mathcal{A})$ is its winning probability.
  \item \textbf{Unused-slot secrecy ($\mathrm{TA}$-$\mathrm{Priv}$).}
    $\mathcal{A}$ announces a challenge slot $\nu^*$ with $|F| \leq t - 1$, never queries $\mathcal{O}_{\mathsf{Approve}}(\cdot, \cdot, \nu^*)$.
    The challenger samples $b \leftarrow \{0,1\}$ and gives $\mathcal{A}$ either $z_0 = k^{(\nu^*)}$ or $z_1 \leftarrow \mathbb{F}_p^2$; $\mathcal{A}$ outputs $b'$, and
    $\Adv^{\mathrm{priv}}(\mathcal{A}) = \bigl| \Pr[b' = 1 \mid b = 0] - \Pr[b' = 1 \mid b = 1] \bigr|$.
  \item \textbf{Attribution ($\mathrm{TA}$-$\mathrm{Attr}$).}
    $\mathcal{A}$ wins if some accepted submission contains a verifying envelope attributed to some $i \notin K \cup F$ although $\mathcal{O}_{\mathsf{Approve}}(i, M^*, \nu^*)$ was never queried.
\end{itemize}
\end{definition}

\subsection{Threshold unforgeability (I0)}
\label{subsec:unforgeability}

\begin{theorem}[Threshold unforgeability]
\label{thm:uf}
Under setup assumption S, for every PPT adversary $\mathcal{A}$ against $\mathrm{TA}$-$\mathrm{UF}$ there are reductions $\mathcal{B}_1, \mathcal{B}_2, \mathcal{B}_3$, each with running time close to $\mathcal{A}$'s, such that
\[
  \Adv^{\mathrm{uf}}(\mathcal{A})
  \;\leq\; n \cdot \Adv^{\mathrm{euf\text{-}cma}}_{\mathrm{Sig}}(\mathcal{B}_1)
  \;+\; nB \cdot \Adv^{\mathrm{ou}}_{\Pi}(\mathcal{B}_2)
  \;+\; \Adv^{\mathrm{cr}}_{H}(\mathcal{B}_3).
\]
\end{theorem}

The statement needs no bound on $|F|$: if $|F| \geq t$ then $|G(M, \nu)| \geq t$ for every operation and the game is unwinnable by definition, the intended custody threshold.
The proof (Appendix~\ref{app:proofs}) splits on the corruption grade of a member the forged quorum must contain: an honest member's envelope yields an EUF-CMA forgery or a hash collision, and a key-compromised member's envelope requires forging a commitment opening.

\begin{remark}[Concrete bounds]
\label{rem:concrete}
With the salted-hash profile the OU term is at most $nB\,(q_H\,2^{-\lambda} + (q_H+1)\,2^{-\ell})$, negligible even for large slot batches at $\lambda = \ell = 256$.
The information-theoretic floor is that a blind guess of an opening-consistent share pair succeeds with probability $p^{-2} \approx 2^{-512}$ per attempt.
The end-to-end forgery level is $\min(\text{member-signature EUF-CMA},\ \text{profile security})$.
\end{remark}

\subsection{Threshold secrecy (I1)}
\label{subsec:secrecy}

\begin{theorem}[Unused-slot secrecy]
\label{thm:priv}
Under setup assumption S, for every PPT $\mathcal{A}$ against $\mathrm{TA}$-$\mathrm{Priv}$ with $|F| \leq t-1$ there is a hiding distinguisher $\mathcal{D}$ with comparable running time such that
\[
  \Adv^{\mathrm{priv}}(\mathcal{A}) \;\leq\; 2\,(n - |F|) \cdot \Adv^{\mathrm{hide}}_{\Pi}(\mathcal{D}).
\]
The advantage is $0$ for the perfectly hiding Pedersen profile and at most $2(n-|F|)\,q_H\,2^{-\lambda}$ for the salted-hash profile in the random-oracle model; given the shares alone, ignoring the commitment layer, secrecy is information-theoretic.
\end{theorem}

\begin{remark}[Used and exposed slots]
\label{rem:used-slots}
Secrecy is claimed only for unused slots.
A consumed slot's shares are public by design, and harmlessly so: C4 guarantees no future operation depends on them.
A slot on which an honest member has approved some message is partially exposed, since that member's shares travel in its envelope.
Operators should time slot lifecycles from envelope publication rather than acceptance, and proposal flow should keep at most one live operation per slot.
Theorem~\ref{thm:uf} keeps the exposure confined: computing the slot polynomials yields neither the per-member opening randomness that C2 demands (Remark~\ref{rem:salt}) nor any signature, so the exposed contribution remains protected by the member's signature, escalates to effective approval only when that member's signing key is also stolen, and never affects other slots or any long-term key.
\end{remark}

\subsection{Operation binding (I2) and member attribution (I3)}

\begin{proposition}[Per-operation binding]
\label{prop:binding}
(a)~An accepted authorization for $(M, \nu)$ cannot be retargeted: a submission for $(M', \nu)$ with $M' \neq M$ that reuses any of its envelopes either fails C1 or yields a collision in $H$.
(b)~At most one operation is ever accepted per coefficient slot.
(c)~The binding covers the custody metadata (address, policy id, operation type), which are inputs to $\opbind{M}{\nu}$.
\end{proposition}

\begin{proposition}[Attribution and non-frameability]
\label{prop:attr}
For every PPT $\mathcal{A}$ against $\mathrm{TA}$-$\mathrm{Attr}$,
\[
  \Pr[\mathcal{A} \text{ wins}] \;\leq\; n \cdot \Adv^{\mathrm{euf\text{-}cma}}_{\mathrm{Sig}} + \Adv^{\mathrm{cr}}_{H}.
\]
\end{proposition}

\subsection{Composition}

\begin{corollary}[Custody authorization guarantees]
\label{prop:custody-guarantees}
Let $\mathrm{Sig}$ be any EUF-CMA signature scheme and $\Pi$ a hiding, opening-unforgeable share-correctness profile, and let the enforcement layer satisfy the integrity assumption of Section~\ref{subsec:model}.
Then the construction of Section~\ref{sec:construction} realizes the custody lifecycle of Section~\ref{subsec:five-ops} with guarantees I0--I3 and the concrete bounds of Theorems~\ref{thm:uf} and~\ref{thm:priv} and Propositions~\ref{prop:binding} and~\ref{prop:attr}, and maintenance can refresh coefficient shares or rotate member keys without changing the authorization logic.
$\mathrm{Sig}$ appears in the bounds only through its EUF-CMA advantage, so it is a deployment parameter rather than a protocol parameter.
\end{corollary}

\section{Deferred proofs}
\label{app:proofs}

\subsection*{Proof of Theorem~\ref{thm:uf} (threshold unforgeability)}

\begin{proof}
Condition on a win: some $\mathcal{O}_{\mathsf{Accept}}$ query accepted $(M^*, \nu^*)$ with a quorum $Q$ of at least $t$ distinct indices, all of C1--C4 passing, while $|G(M^*, \nu^*)| < t$ at that moment.
Since $|Q| \geq t > |G(M^*, \nu^*)|$, there is an index $i^\dagger \in Q \setminus G(M^*, \nu^*)$; fix the smallest.
As $F \subseteq G(M^*, \nu^*)$ we have $i^\dagger \notin F$, leaving two exhaustive cases.
Write $\varepsilon_{\mathrm{sig}}$ and $\varepsilon_{\mathrm{open}}$ for the probabilities of a win with $i^\dagger \notin K$ and $i^\dagger \in K$ respectively, so $\Adv^{\mathrm{uf}}(\mathcal{A}) \leq \varepsilon_{\mathrm{sig}} + \varepsilon_{\mathrm{open}}$.

\emph{Case 1: $i^\dagger \notin K \cup F$ (signature forgery or hash collision).}
By C1 the submission carries a valid signature $\text{sig}_{i^\dagger}$ under $\text{pk}_{i^\dagger}$ on an envelope $A_{i^\dagger}$ whose fields include $(\opbind{M^*}{\nu^*}, \coeffid_{\nu^*}, i^\dagger)$.
Honest $C_{i^\dagger}$ signed only the envelopes returned by $\mathcal{O}_{\mathsf{Approve}}(i^\dagger, \cdot, \cdot)$, and $i^\dagger \notin G(M^*, \nu^*)$ means the query $(i^\dagger, M^*, \nu^*)$ never occurred.
The envelope encoding is canonical and injective, so $A_{i^\dagger}$ coincides with a previously signed envelope only if that envelope was produced for some $(M, \nu)$ with $\coeffid_\nu = \coeffid_{\nu^*}$ (hence $\nu = \nu^*$) and $\opbind{M}{\nu^*} = \opbind{M^*}{\nu^*}$ despite $M \neq M^*$; the latter exhibits distinct inputs with equal outputs in the computation of $\opbind{M}{\nu^*}$, and $\mathcal{B}_3$ outputs the colliding pair.
Otherwise $A_{i^\dagger}$ was never signed by the oracle and $(A_{i^\dagger}, \text{sig}_{i^\dagger})$ is an EUF-CMA forgery.
Reduction $\mathcal{B}_1$ picks $j$ uniformly in $[n] \setminus (K \cup F)$, installs the EUF-CMA challenge key as $\text{pk}_j$, runs the rest of setup itself, answers $\mathcal{O}_{\mathsf{Approve}}(j, \cdot, \cdot)$ through its signing oracle, and on a win with $i^\dagger = j$ outputs the forgery.
So $\varepsilon_{\mathrm{sig}} \leq n \cdot \Adv^{\mathrm{euf\text{-}cma}}_{\mathrm{Sig}}(\mathcal{B}_1) + \Adv^{\mathrm{cr}}_{H}(\mathcal{B}_3)$.

\emph{Case 2: $i^\dagger \in K \setminus G(M^*, \nu^*)$ (opening forgery).}
By Definition~\ref{def:effective} no query $\mathcal{O}_{\mathsf{Approve}}(i^\dagger, \cdot, \nu^*)$ occurred, so the opening randomness $\rho_{i^\dagger}^{(\nu^*)}$ was never released.
The adversary may well know the shares $s_{i^\dagger}^{(\nu^*)}$ themselves, but C2 accepted an opening $(s', \rho')$ with $\Vf(\text{com}_{i^\dagger}^{(\nu^*)}, s', \rho') = 1$, and producing \emph{any} valid opening without the dealt randomness is exactly the OU game of Definition~\ref{def:profile}.
Reduction $\mathcal{B}_2$ guesses $(j, \mu) \leftarrow K \times [B]$, runs setup itself, submits $m = s_j^{(\mu)}$ to the OU challenger, and installs the returned commitment as $\text{com}_j^{(\mu)}$.
It can answer every oracle query: the only value it lacks is $\rho_j^{(\mu)}$, needed only for $\mathcal{O}_{\mathsf{Approve}}(j, \cdot, \mu)$, and if that query arrives then $j \in G(\cdot, \mu)$ and the guess was already wrong, so $\mathcal{B}_2$ aborts.
In the salted-hash instantiation the random oracle requires no simulation: $H$ is a public oracle shared by $\mathcal{A}$, $\mathcal{B}_2$, and the OU challenger, and $\mathcal{B}_2$ only observes it, so the embedding is exact.
In particular, an adversary that queries $H$ on the challenge preimage $\text{``custody-commit''} \,\|\, \rho_j^{(\mu)} \,\|\, s_j^{(\mu)}$, thereby learning the withheld salt, has not evaded the reduction; it has produced a valid opening, the event already charged inside $\Adv^{\mathrm{ou}}_{\Pi}$ by the $q_H\,2^{-\lambda}$ term of Lemma~\ref{lem:profiles}.
On a win with a correct guess it outputs $(s', \rho')$, so
$\varepsilon_{\mathrm{open}} \leq nB \cdot \Adv^{\mathrm{ou}}_{\Pi}(\mathcal{B}_2)$.
\end{proof}

\subsection*{Proof of Theorem~\ref{thm:priv} (unused-slot secrecy)}

\begin{proof}
Let $h = n - |F|$ and let $i_1 < \dots < i_h$ enumerate $[n] \setminus F$.
Consider hybrids $\mathsf{H}_0, \dots, \mathsf{H}_h$, where $\mathsf{H}_0$ is the $b = 0$ experiment and $\mathsf{H}_j$ replaces the challenge-slot commitments $\text{com}_{i_1}^{(\nu^*)}, \dots, \text{com}_{i_j}^{(\nu^*)}$ by commitments to $(0,0)$ under fresh randomness.
Adjacent hybrids are simulatable from a hiding challenge: the reduction samples all setup secrets itself, submits $m_0 = s_{i_j}^{(\nu^*)}$ and $m_1 = (0,0)$, and embeds the challenge commitment as $\text{com}_{i_j}^{(\nu^*)}$.
It never needs the challenge opening: approvals on $\nu^*$ are disallowed by the experiment, and $\mathcal{O}_{\mathsf{Accept}}$ only runs the public $\Vf$ on adversary-supplied openings.
Since the reduction knows $k^{(\nu^*)}$ it can hand over the challenge value for either bit, so $|\Pr[\mathsf{H}_{j-1} = 1] - \Pr[\mathsf{H}_j = 1]| \leq \Adv^{\mathrm{hide}}_{\Pi}(\mathcal{D})$.

In $\mathsf{H}_h$ the view of $\mathcal{A}$ before the challenge depends on slot $\nu^*$ only through the shares $\{s_i^{(\nu^*)}\}_{i \in F}$ delivered at setup: the $\nu^*$-commitments are now independent of the slot, no $\nu^*$-envelope exists, and all other slots use independent polynomials (assumption S).
These are at most $t-1$ evaluations of each of two degree-$(t-1)$ polynomials, so by perfect secrecy of Shamir sharing they are jointly independent of the uniform $k^{(\nu^*)}$.
Hence in $\mathsf{H}_h$ the real challenge $z_0 = k^{(\nu^*)}$ is uniform and independent of the view, exactly the distribution of $z_1$, and the $b=0$ and $b=1$ experiments coincide.
Walking the same hybrids back inside the $b = 1$ experiment gives the factor $2$.
\end{proof}

\subsection*{Proof of Proposition~\ref{prop:binding} (per-operation binding)}

\begin{proof}
(a)~Every envelope fixes $\opbind{M}{\nu}$ under the member's signature; acceptance for $M'$ recomputes $\opbind{M'}{\nu}$ from the submitted operation and C1 requires the signed field to equal it, so a reused envelope forces $\opbind{M'}{\nu} = \opbind{M}{\nu}$, i.e.\ distinct inputs with equal $H$-outputs.
(b)~is the C4 state machine: $\coeffid_\nu$ enters the monotone consumed set at the first acceptance, and the enforcement layer evaluates C4 atomically in finality order, so a second acceptance is impossible unconditionally.
(c)~is immediate from the definition of $\opbind{M}{\nu}$.
\end{proof}

\subsection*{Proof of Proposition~\ref{prop:attr} (attribution and non-frameability)}

\begin{proof}
A win is precisely Case 1 of Theorem~\ref{thm:uf} for the framed member: a verifying envelope under $\text{pk}_i$ with $i \notin K \cup F$ for an operation $(M^*, \nu^*)$ never queried for $i$; the reductions $\mathcal{B}_1$ and $\mathcal{B}_3$ apply verbatim.
Conversely, every accepted contribution names its member index inside the signed envelope, so accepted quorums are attributable by construction, and relay substitution, share swapping, and impersonation all invalidate $\text{sig}_i$.
\end{proof}

\end{document}